\documentclass[final]{l4dc2024}

\title[Physically Consistent Modeling \& Identification of Nonlinear Friction with D-
GP]{Physically Consistent Modeling \& Identification of Nonlinear Friction with Dissipative
Gaussian Processes}
\usepackage{times}
\usepackage{algorithm}
\usepackage{algpseudocode}
\usepackage{bm}

\author{
 \Name{Rui Dai}$^{\dagger, *}$ \Email{rui.dai@iit.it}\\
 \Name{Giulio Evangelisti}$^{*}$ \Email{giulio.evangelisti@tum.de}\\
 \Name{Sandra Hirche}$^{*}$ \Email{hirche@tum.de} \\
 \addr $^{\dagger}$Humanoids and Human-Centered Mechatronics Research Line, Istituto Italiano di Tecnologia (IIT), 16163 Genova, Italy. \\
 \addr $^{*}$Chair of Information-oriented Control, Department of Electrical and Computer Engineering, Technical University of Munich, 80333 Munich, Germany
}

\newtheorem{theorem_}{Theorem}
\newtheorem{assumption}{Assumption}
\newtheorem{corollary_}{Corollary}

\begin{document}

\maketitle

\begin{abstract}%
Friction modeling has always been a challenging problem due to the complexity of real physical systems. Although a few state-of-the-art structured data-driven methods show their efficiency in nonlinear system modeling, deterministic passivity as one of the significant characteristics of friction is rarely considered in these methods.
To address this issue, we propose a Gaussian Process based model that preserves the inherent structural properties such as passivity. A matrix-vector physical structure is considered in our approaches to ensure physical consistency, in particular, enabling a guarantee of positive semi-definiteness of the damping matrix.
An aircraft benchmark simulation is employed to demonstrate the efficacy of our methodology. Estimation accuracy and data efficiency are increased substantially by considering and enforcing more structured physical knowledge.
Also, the fulfillment of the dissipative nature of the aerodynamics is validated numerically.
\end{abstract}


\section{Introduction}
Passive systems are a class of dynamical systems
that can only increase stored energy when there is external energy inflow.
This property holds practical significance in Euler-Lagrangian systems, including robots, aircraft, vehicles, etc.
For mechanical systems, 
friction is a direct source of dissipation.
Modeling with a guarantee of its passivity can help ensure the system's stability by the passivity theorem \citep{khalil_nonlinear_2002}.
In addition, its accurate modeling is also beneficial in improving control accuracy by friction compensation \citep{Olsson_98}.

For decades, researchers have constructed many analytical models of friction that can characterize specific characteristics, from static to dynamic \citep{Bliman_1995, de_1995, Lampaert_2002, FENG_2019}. However, the characteristics are usually nonlinear, and the behavior of friction may vary significantly under different conditions, such as hysteretic effects \citep{Hysteretic_2008}.
There is no universal physical model that can be applied to all situations.
As a promising alternative, several intelligent model-free computing approaches can capture the complexity and nonlinear properties of a system by analyzing observations from the environments \citep{Huang_2019, Peng_Mechanistic_2022, Auriol_Drill_2022, Dong_Friction_2021}. However, modeling the dynamics from data is often challenging because of the high dimensionality of the real-world systems \citep{Geist_2021}. In addition, data collection on physical systems is expensive and time-consuming.
To address this issue, recent efforts strive to integrate physical structure as prior knowledge with data-driven modeling to improve learning efficiency.

\subsection{Related Work}
The structured models inherit the potential to combine the advantages of both analytical and data-driven modeling while avoiding their limitations when used individually \citep{Geist_2021, Nelles_2020}. Rigid-body dynamics represent a valuable source of prior structural knowledge for mechanical systems
\citep{Jemin_2019}. 
\cite{Nguyen_2010} introduce two semiparametric regression approaches, where the physical model knowledge can be incorporated into either the mean function or the kernel of a nonparametric Gaussian Process Regression (GPR). Their findings demonstrate that semiparametric models, trained with rigid body dynamics as a prior, outperform standard rigid body dynamics models when applied to real-world data. Physical constraints are integrated into the learning process by introducing structure to the learned parameters in \cite{Giovanni_2020}. This framework enables the acquisition of physically plausible dynamics through gradient descent, enhancing both training speed and the generalization of learned dynamic models.
In addition, structured learning frameworks can guarantee physical properties such as energy conservation of Lagrangian systems \citep{Evangelisti_Physically_2022}. This Lagrangian-GP-based learning framework is extended to include dissipative systems in \cite{Evagelisti_2024_Data-Driven} but a passivity preserving guarantee is not given.
In the literature, passivity as a significant phenomenon in mechanical systems is rarely considered.
\cite{Koch_2020_Verifying} verifies dissipativity of linear time-invariant systems from noisy data with guaranteed robustness and the sector bounds which can probabilistically show the passivity is derived in \cite{Fiedler_2021_learning}. However, to the best knowledge of the authors, deterministic guarantee of passivity from data for nonlinear systems has not been investigated.

\subsection{Contribution and Structure}
The main contribution of our work is proposing a physically structured Gaussian Process that preserves passivity for friction modelling and identification.
We consider a matrix-vector structure for the prior mean and kernels in GP, with a full and diagonal damping matrix, respectively. We investigate their passivity guarantee of the damping estimates by showing the positive semi-definiteness of the damping matrix. 
 
After introducing the formulation of our problem in Section~\ref{sec: Problem Formulation and Passivity}, we review the basic GP framework in Section~\ref{sec: Gaussian Process Framework}.
In Section~\ref{sec: Dissipative Gaussian Processes}, structured dissipative GP that provide passivity guarantees are proposed. 
We evaluate our approach and discuss the results in Section~\ref{sec: Simulation and Results} and conclude in Section~\ref{sec: Conclusion}.

\subsection{Notation}
We use bold lower/upper case characters to denote vectors/matrices, respectively. $\boldsymbol{I}$ denotes the identity. 
$[\boldsymbol{x}_i]_{i=1}^\mathcal{D} \in \mathbb{R}^{N\mathcal{D}}$ denotes a stacked vector of $\mathcal{D}$ indexed vectors $\boldsymbol{x}_i \in \mathbb{R}^{N}$ and $[\boldsymbol{X}_i]_{i=1}^\mathcal{D} \in \mathbb{R}^{N\mathcal{D} \times N}$ a stacked matrix of $\mathcal{D}$ indexed matrices $\boldsymbol{X}_i \in \mathbb{R}^{N \times N}$. 
The direct sum is denoted as $\oplus$. We use $\succeq$ and $\preceq$ to denote matrix inequality.

\section{Problem Formulation and Passivity}
\label{sec: Problem Formulation and Passivity}

Our work considers Euler-Lagrange systems \citep{Ortega_1998} whose equations of motion are given by

\begin{equation}
\boldsymbol{M}(\boldsymbol{q}) \ddot{\boldsymbol{q}}+\boldsymbol{C}(\boldsymbol{q}, \dot{\boldsymbol{q}}) \dot{\boldsymbol{q}}+\boldsymbol{g}(\boldsymbol{q})+\boldsymbol{D}(\dot{\boldsymbol{q}}) \dot{\boldsymbol{q}}=\boldsymbol{\tau},
\label{eq: EL motion equation}
\end{equation}
where $\boldsymbol{q} \in \mathbb{R}^{N}$ is the generalized coordinate, $\boldsymbol{M}(\boldsymbol{q}) \in \mathbb{R}^{N \times N}$ the positive definite inertia matrix, $\boldsymbol{C}(\boldsymbol{q}, \dot{\boldsymbol{q}}) \in \mathbb{R}^{N \times N}$ the Coriolis matrix, $\boldsymbol{g}(\boldsymbol{q}) \in \mathbb{R}^{N}$ the generalized potential forces and $\boldsymbol{D}(\dot{\boldsymbol{q}}) \in \mathbb{R}^{N \times N}$ the damping matrix. Note that we neglect the positional dependence of the damping matrix to prevent notational complexity. $\boldsymbol{\tau} \in \mathbb{R}^{N}$ is the 
input actuation applied to the system.
This formulation can be applied to a large varieties of rigid body systems including robots, aircraft, etc.
According to the definition of passivity in \cite{Willems_1972},
passivity of this system can be guaranteed if the damping matrix $\boldsymbol{D}(\dot{\boldsymbol{q}})$ is positive semidefinite, i.e. $\dot{\boldsymbol{q}}^{\top}\boldsymbol{D}(\dot{\boldsymbol{q}}) \dot{\boldsymbol{q}} \geq 0$.

\section{Gaussian Process Framework}
\label{sec: Gaussian Process Framework}
We first review the Gaussian Process regression framework based on \cite{rasmussen_gaussian_2006} in this section. A Gaussian process is a collection of random variables and any finite number of which have a joint Gaussian distribution. It can be used to describe a distribution over functions and consider inference in a function space.

Considering $N$-dimensional inputs $\boldsymbol{x}, \boldsymbol{x}^{\prime} \in \mathcal{X}$ in a continuous domain $\mathcal{X} \subseteq \mathbb{R}^{N}$, and a real process $f(\boldsymbol{x}) \in \mathbb{R}$, a GP is completely specified by its mean $m(\boldsymbol{x})=\mathbb{E}[f(\boldsymbol{x})]$ and covariance function $k\left(\boldsymbol{x}, \boldsymbol{x}^{\prime}\right)=\mathbb{E}\left[(f(\boldsymbol{x})-m(\boldsymbol{x}))\left(f\left(\boldsymbol{x}^{\prime}\right)-m\left(\boldsymbol{x}^{\prime}\right)\right)\right]$, written as
$
f(\boldsymbol{x}) \sim \mathcal{G} \mathcal{P}\left(m(\boldsymbol{x}), k(\boldsymbol{x}, \boldsymbol{x}^{\prime})\right).
$
The mean function provides the prior knowledge of the unknown function, while the kernel determines higher-level functional properties such as smoothness. It specifies the similarity between two values of the function evaluated on two specific inputs.

The GP inherits the Gaussian distribution properties
and we can infer the function from the given observations.
Considering $\mathcal{D}$ finite noisy observations $\boldsymbol{y} = [y_d]_{d=1}^{\mathcal{D}}$ with $y_d = f(\boldsymbol{x}_d) + \varepsilon$, where $\varepsilon$ is an additive i.i.d. Gaussian noise process with variance $\sigma_{\varepsilon}^2$, then the covariance of the noisy observations is
$
\operatorname{cov}(\boldsymbol{y}) = \boldsymbol{K}(\boldsymbol{X}, \boldsymbol{X}) + \boldsymbol{\Sigma} = \boldsymbol{K}(\boldsymbol{X}, \boldsymbol{X}) + \sigma_{\varepsilon}^2 \boldsymbol{I},
$
where $\boldsymbol{K}(\boldsymbol{X}, \boldsymbol{X}) \in \mathbb{R}^{\mathcal{D} \times \mathcal{D}}$
denotes the covariance evaluated at all pairs of training points. The joint Gaussian distribution of the observations and the test points can be written as

\begin{equation}
\left[\begin{array}{c}
\boldsymbol{y} \\
{\boldsymbol{f}(\boldsymbol{X}_*)}
\end{array}\right] \sim \mathcal{N}\left(\left[\begin{array}{c}\boldsymbol{m}(\boldsymbol{X}) \\ \boldsymbol{m}(\boldsymbol{X}_*)\end{array}\right],\left[\begin{array}{cc}
\boldsymbol{K}(\boldsymbol{X}, \boldsymbol{X})+\sigma_{\varepsilon}^2 \boldsymbol{I} & \boldsymbol{K}\left(\boldsymbol{X}, \boldsymbol{X}_*\right) \\
\boldsymbol{K}\left(\boldsymbol{X}_*, \boldsymbol{X}\right) & \boldsymbol{K}\left(\boldsymbol{X}_*, \boldsymbol{X}_*\right)
\end{array}\right]\right),
\end{equation}
where $\boldsymbol{m}(\cdot)$ denotes the prior mean of the function, $\boldsymbol{K}\left(\boldsymbol{X}, \boldsymbol{X}_*\right)$ denotes the covariance evaluated at all pairs of training data and test data, $\boldsymbol{K}\left(\boldsymbol{X}_*, \boldsymbol{X}_*\right)$ denotes the covariance evaluated at all pairs of test data and $\boldsymbol{K}\left(\boldsymbol{X}_*, \boldsymbol{X}\right) = \boldsymbol{K}^{\top}\left(\boldsymbol{X}, \boldsymbol{X}_*\right)$. 
The posterior conditional distribution based on the prior distribution on the observations can be obtained by
$
\boldsymbol{f}(\boldsymbol{X}_*) \mid \boldsymbol{X}, \boldsymbol{y}, \boldsymbol{X}_* \sim \mathcal{N}\left(\overline{\boldsymbol{f}}(\boldsymbol{X_*}), \operatorname{cov}\left(f(\boldsymbol{X_*})\right)\right)
$,
where
$
\overline{\boldsymbol{f}}(\boldsymbol{X}_*) 
=\boldsymbol{m}(\boldsymbol{X}_*) + \boldsymbol{K}\left(\boldsymbol{X}_*, \boldsymbol{X}\right)\left[\boldsymbol{K}(\boldsymbol{X}, \boldsymbol{X})+\sigma_{\varepsilon}^2 \boldsymbol{I}\right]^{-1} (\boldsymbol{y} - \boldsymbol{m}(\boldsymbol{X})),
$
$
\operatorname{cov}\left(f(\boldsymbol{X_*})\right) =\boldsymbol{K}\left(\boldsymbol{X}_*, \boldsymbol{X}_*\right)-\boldsymbol{K}\left(\boldsymbol{X}_*, \boldsymbol{X}\right)\left[\boldsymbol{K}(\boldsymbol{X}, \boldsymbol{X})+\sigma_{\varepsilon}^2 \boldsymbol{I}\right]^{-1} \boldsymbol{K}\left(\boldsymbol{X}, \boldsymbol{X}_*\right).
\label{eq: gp regression}
$

\section{Dissipative Gaussian Processes}
\label{sec: Dissipative Gaussian Processes}

After introducing the passive system and the basic GP framework, we move on to our contribution of a physically consistent GP enforcing a damping torque with a matrix-vector form. 
Passivity is guaranteed analytically via a sufficient condition for ensuring positive semidefiniteness of the damping matrix.
For this, we analyze the posterior mean estimates of a 
full, and then diagonal, damping matrix.

\subsection{Matrix-vector Structured Gaussian Processes for Full Damping Matrices}
\label{sec: Damping matrix with Full Elements}
 
In this section, we first
investigate the general case of a full damping matrix.

\subsubsection{GP Structure and Formulation}
\label{sec: GP Structure and Formulation}

For the unknown damping term in the physical dynamic system~\eqref{eq: EL motion equation}, we assume the independence of the components in the damping matrix.

\begin{assumption}
    Consider a damping torque vector in the form $\boldsymbol{\tau}^d(\Dot{\boldsymbol{q}})=\boldsymbol{D}(\Dot{\boldsymbol{q}})\Dot{\boldsymbol{q}}$, where $\Dot{\boldsymbol{q}} \in \mathbb{R}^N$ is the generalized velocity. The elements $d_{mn}(\dot{\bm{q}})$ of the damping matrix $\boldsymbol{D}(\Dot{\boldsymbol{q}}) = [d_{mn}(\Dot{\boldsymbol{q}})]_{m,n = 1}^N \in \mathbb{R}^{N \times N}$ are 
    stochastically independent, i.e., $\mathbb{E}[d_{mn}d_{ij}]=\mathbb{E}[d_{mn}]\mathbb{E}[d_{ij}]$ for $m,n\neq i,j$.
    \label{ass: full independent damping matrix}
\end{assumption}
This assumption is reasonable because we can choose a suitable input coordinate such that each dimension is independent and the torque is in its span.
We model each element $d_{mn}(\Dot{\boldsymbol{q}})$ independently
\begin{equation}
d_{mn}(\Dot{\boldsymbol{q}}) \sim \mathcal{G} \mathcal{P}\left(m_{d_{mn}}(\Dot{\boldsymbol{q}}), k_{d_{mn}}\left(\Dot{\boldsymbol{q}}, \Dot{\boldsymbol{q}}^{\prime}\right)\right)
\label{eq: full D element GP model} \, .
\end{equation}

\begin{lemma}
    Under Ass.~\ref{ass: full independent damping matrix} and the independent GP models~\eqref{eq: full D element GP model}, denoting $\boldsymbol{k}_{d_n} = [k_{d_{mn}}\left(\Dot{\boldsymbol{q}}, \Dot{\boldsymbol{q}}^{\prime}\right)]_{m=1}^{N}$ as the n-th column of the matrix kernel associated with $\boldsymbol{D}(\Dot{\boldsymbol{q}})$, the covariance of the damping torque $\boldsymbol{\tau}^d(\Dot{\boldsymbol{q}})$ is
    \begin{equation}
        \boldsymbol{K}_{\boldsymbol{\tau}}(\Dot{\boldsymbol{q}}, \Dot{\boldsymbol{q}}^{\prime} ) = \operatorname{cov}(\boldsymbol{\tau}^d(\Dot{\boldsymbol{q}})) = 
        \sum_{n=1}^{N} \Dot{q}_n \Dot{q}_n^{\prime} \mathrm{diag}(\boldsymbol{k}_{d_n})
    \label{eq: full Ktau}
    \end{equation}
\end{lemma}
\begin{proof}
    Define a mean matrix $\bm{M_D}(\Dot{\boldsymbol{q}}) = [m_{d_{mn}}(\Dot{\boldsymbol{q}})]_{m,n = 1}^N$  corresponding to $\boldsymbol{D}(\Dot{\boldsymbol{q}})$, then
    \begin{equation}
    \begin{aligned}
        \operatorname{cov}(\boldsymbol{\tau}^d(\Dot{\boldsymbol{q}})) &= \mathbb{E}\left[\left(\boldsymbol{D}(\Dot{\boldsymbol{q}})\Dot{\boldsymbol{q}} - \bm{M_D}(\Dot{\boldsymbol{q}})\Dot{\boldsymbol{q}}\right)\left(\boldsymbol{D}(\Dot{\boldsymbol{q}}^{\prime})\Dot{\boldsymbol{q}}^{\prime} - \bm{M_D}(\Dot{\boldsymbol{q}}^{\prime})\Dot{\boldsymbol{q}}^{\prime}\right)^{\top}\right] \\
        &= \mathbb{E}\left[\left(\sum_{n=1}^{N}\Dot{q}_n\left(\boldsymbol{d}_{\cdot n}(\Dot{\boldsymbol{q}})-\boldsymbol{m}_{\boldsymbol{d}_{\cdot n}}(\Dot{\boldsymbol{q}})\right)\right)\left(\sum_{m=1}^{N}\Dot{q}_m^{\prime}\left(\boldsymbol{d}_{\cdot m}(\Dot{\boldsymbol{q}})-\boldsymbol{m}_{\boldsymbol{d}_{\cdot m}}(\Dot{\boldsymbol{q}})\right)\right)^{\top}\right] \\
        &= \sum_{n=1}^{N}\sum_{m=1}^{N}\Dot{q}_n\Dot{q}_m^{\prime}\mathbb{E}\left[\left(\boldsymbol{d}_{\cdot n}(\Dot{\boldsymbol{q}})-\boldsymbol{m}_{\boldsymbol{d}_{\cdot n}}(\Dot{\boldsymbol{q}})\right)\left(\boldsymbol{d}_{\cdot m}(\Dot{\boldsymbol{q}})-\boldsymbol{m}_{\boldsymbol{d}_{\cdot m}}(\Dot{\boldsymbol{q}})\right)^{\top}\right] \\
        &= \sum_{n=1}^{N}\Dot{q}_n\Dot{q}_n^{\prime}\mathbb{E}\left[\left(\boldsymbol{d}_{\cdot n}(\Dot{\boldsymbol{q}})-\boldsymbol{m}_{\boldsymbol{d}_{\cdot n}}(\Dot{\boldsymbol{q}})\right)\left(\boldsymbol{d}_{\cdot n}(\Dot{\boldsymbol{q}})-\boldsymbol{m}_{\boldsymbol{d}_{\cdot n}}(\Dot{\boldsymbol{q}})\right)^{\top}\right] \\
        & = \sum_{n=1}^{N} \Dot{q}_n \Dot{q}_n^{\prime} \mathrm{diag}(\boldsymbol{k}_{d_n})
    \end{aligned}
    \end{equation}
\end{proof}
Let us model the torques as a GP
\begin{equation}
    \boldsymbol{\tau}^d(\Dot{\boldsymbol{q}}) \sim \mathcal{G} \mathcal{P}\left(\boldsymbol{m}_{\boldsymbol{\tau}}(\Dot{\boldsymbol{q}}), \boldsymbol{K}_{\boldsymbol{\tau}}(\Dot{\boldsymbol{q}}, \Dot{\boldsymbol{q}}^{\prime} )\right),
    \label{eq: Diag tau GP}
\end{equation}
with the structured prior mean $\boldsymbol{m}_{\boldsymbol{\tau}}(\Dot{\boldsymbol{q}})=\bm{M_D}(\Dot{\boldsymbol{q}})\Dot{\boldsymbol{q}}$ and kernel~\eqref{eq: full Ktau}.
Assume access to a number of $\mathcal{D}$ noise-free generalized velocity measurements $\Dot{\boldsymbol{Q}} = [\Dot{\boldsymbol{q}}_i^{\top}]_{i=1}^{\mathcal{D}}$ and noise-corrupted output measurements $\boldsymbol{Y} = [\boldsymbol{y}_i^{\top}]_{i=1}^{\mathcal{D}}$, where 
$
    \boldsymbol{y}_i = \boldsymbol{\tau}_{i}^d(\Dot{\boldsymbol{q}}_i) + \boldsymbol{\varepsilon}_i, \; \boldsymbol{\varepsilon}_i \sim \mathcal{N}(\boldsymbol{0}, \sigma_{\varepsilon}^2 \boldsymbol{I}).
$
Introducing a stacked vector of outputs $\boldsymbol{y} = \operatorname{vec}(\boldsymbol{Y}^{\top})$ with mean $\boldsymbol{m}_{\boldsymbol{y}} = [\boldsymbol{m}_{\boldsymbol{\tau}}(\Dot{\boldsymbol{q}}_i)]_{i=1}^{\mathcal{D}}$, then
the joint distribution of observations and desired estimate with velocity $\dot{\boldsymbol{q}}_*$ is given by

\begin{equation}
\left[\begin{array}{c}
\boldsymbol{y} \\
\boldsymbol{\tau}^d(\Dot{\boldsymbol{q}}_*)
\end{array}\right] \sim \mathcal{N}\left(\left[\begin{array}{c}
\boldsymbol{m_y} \\
\boldsymbol{m_{\tau}}(\Dot{\boldsymbol{q}}_*)
\end{array}\right],\left[\begin{array}{cc}
\boldsymbol{K_y}(\Dot{\boldsymbol{Q}}, \Dot{\boldsymbol{Q}}) & \boldsymbol{K_{\tau}}(\Dot{\boldsymbol{Q}}, \Dot{\boldsymbol{q}}_*) \\
\boldsymbol{K_{\tau}}^{\top}(\Dot{\boldsymbol{Q}}, \Dot{\boldsymbol{q}}_*) & \boldsymbol{K_\tau}(\Dot{\boldsymbol{q}}_*, \Dot{\boldsymbol{q}}_*),
\end{array}\right]\right),
\label{eq: Diag-D-GP}
\end{equation}
where
\begin{subequations}
\begin{align}
    \boldsymbol{K_y}(\Dot{\boldsymbol{Q}}, \Dot{\boldsymbol{Q}}) &= [\boldsymbol{K_{\tau}}(\Dot{\boldsymbol{q}}_i, \Dot{\boldsymbol{q}}_j)]_{i,j=1}^\mathcal{D} + \sigma_{\varepsilon}^2 \boldsymbol{I} \\
    \boldsymbol{K_{\tau}}(\Dot{\boldsymbol{Q}}, \Dot{\boldsymbol{q}}_*) &= [\boldsymbol{K_{\tau}}(\Dot{\boldsymbol{q}}_i, \Dot{\boldsymbol{q}}_*)]_{i=1}^\mathcal{D}.
\end{align}
\end{subequations}

\subsubsection{Positive Semidefiniteness}
\label{sec: Positive Definiteness}

Before we show the guarantee for passivity of this specific structured GP, two minor restrictions on the used kernel functions and the structure of the prior mean need to be imposed.
\begin{assumption}
    Each 
    kernel component 
    is bounded by $\sigma_{f_{mn}}^2$, respectively, i.e. $|k_{d_{mn}}\left(\Dot{\boldsymbol{q}}, \Dot{\boldsymbol{q}}^{\prime}\right)| \leq \sigma_{f_{mn}}^2$.
\label{ass: class M full}
\end{assumption}

\begin{assumption}
    The prior mean matrix $\bm{M_D}(\Dot{\boldsymbol{q}})$ is diagonal and positive semi-definite, i.e. 
    $
        \boldsymbol{m}_{\boldsymbol{\tau}}(\Dot{\boldsymbol{q}}) = \operatorname{diag}(\bm{m_d}(\dot{\bm{q}}))\Dot{\boldsymbol{q}}
    $
    with non-negative ${m}_{{d}_m}(\Dot{\boldsymbol{q}})\geq0$.
    \label{ass: prior mean}
\end{assumption}
Note that Assumption~\ref{ass: class M full} includes a wide variety of covariance functions including the commonly used squared exponential (SE) kernel and Assumption~\ref{ass: prior mean} can be easily constructed by the prior mean.
Then, according to the conditional distribution of the mean estimate in Section~\ref{sec: Gaussian Process Framework}, we obtain
    \begin{equation}
    \begin{aligned}
        \hat{\boldsymbol{\tau}}^d(\Dot{\boldsymbol{q}}) 
        &= \operatorname{diag}(\bm{m_d}(\dot{\bm{q}}))\dot{\bm{q}} + \sum_{n=1}^{N} \Dot{q}_n \boldsymbol{K}^{\top}_{n}(\Dot{\boldsymbol{Q}}, \Dot{\boldsymbol{q}})\Tilde{\dot{\boldsymbol{Q}}}_n\boldsymbol{K}_{\boldsymbol{y}}^{-1}(\Dot{\boldsymbol{Q}}, \Dot{\boldsymbol{Q}})(\boldsymbol{y} - \boldsymbol{m_y}),
    \label{eq: full tau m hat}
    \end{aligned}
    \end{equation}
where $
        \boldsymbol{K}_{{n}}(\Dot{\boldsymbol{Q}}, \Dot{\boldsymbol{q}}) = [\boldsymbol{K}_{{n}}(\Dot{\boldsymbol{q}}_i, \Dot{\boldsymbol{q}})]_{i=1}^{D}
    $ with
    $
        \boldsymbol{K}_{{n}}(\Dot{\boldsymbol{q}}_i, \Dot{\boldsymbol{q}}) = \operatorname{diag}(\boldsymbol{k}_{d_{n}}(\Dot{\boldsymbol{q}}_i, \Dot{\boldsymbol{q}}))
    $.
    $\Tilde{\dot{\boldsymbol{Q}}}_n = \oplus_{i=1}^{\mathcal{D}} (\Tilde{\dot{{q}}}_{n,i})\boldsymbol{I}$ with $\boldsymbol{I} \in \mathbb{R}^{{N \times N}}$, 
and $\Tilde{\dot{{q}}}_{n,i}$ is the $(i, n)$-th element in ${\dot{\boldsymbol{Q}}}$. We also denote $\Delta \boldsymbol{y} = \boldsymbol{y} - \boldsymbol{m_y}$ and $\Tilde{\dot{{\boldsymbol{q}}}} = \operatorname{vec}({\dot{\boldsymbol{Q}}}^{\top})$.

\begin{theorem_}
    Consider
    the damping torque estimate~\eqref{eq: full tau m hat}
    under Ass.~\ref{ass: class M full} and~\ref{ass: prior mean}.
    Positive semidefiniteness of the matrix estimate $\hat{\boldsymbol{D}}(\Dot{\boldsymbol{q}})$ is guaranteed deterministically if  
    \begin{equation}
    \begin{aligned}
        \boldsymbol{\Sigma}_f & \preceq \frac{\sigma_{\varepsilon}^2}{\sqrt{\mathcal{D}} \|\Tilde{\dot{\boldsymbol{q}}}\|_{\infty} \left\Vert\Delta \boldsymbol{y} \right\Vert} \mathrm{diag}(\bm{m_d}(\dot{\bm{q}}))
    \label{eq: full D-GP bound}
    \end{aligned}
    \end{equation}
    \label{th: full bound}
    holds $\forall \dot{\boldsymbol{q}} \in \mathbb{R}^N$ with the hypervariance matrix $\boldsymbol{\Sigma}_f=[\,|\sigma_{f_{mn}}|\,]_{m,n=1}^N$.
\end{theorem_}
\begin{proof}
    According to \eqref{eq: full tau m hat}, we can compute the dissipated power as
    \begin{equation}
    \Dot{\boldsymbol{q}}^{\top}\hat{\boldsymbol{\tau}}^d(\Dot{\boldsymbol{q}}) = \Dot{\boldsymbol{q}}^{\top}\mathrm{diag}(\bm{m_d}(\dot{\bm{q}}))\Dot{\boldsymbol{q}} + \sum_{m=1}^{N} \sum_{n=1}^{N} \Dot{q}_m \Dot{q}_n \underbrace{\boldsymbol{k}^{\top}_{d_{mn}}(\Dot{\boldsymbol{Q}}, \Dot{\boldsymbol{q}})\Tilde{\dot{\boldsymbol{Q}}}_n}_{=\tilde{\bm{k}}_{mn}^{\top}}\underbrace{\boldsymbol{K}_{\boldsymbol{y}}^{-1}(\Dot{\boldsymbol{Q}}, \Dot{\boldsymbol{Q}})\Delta \boldsymbol{y}}_{=\Delta \boldsymbol{x}}.
    \end{equation}
    with $\boldsymbol{k}_{d_{mn}}(\Dot{\boldsymbol{Q}}, \Dot{\boldsymbol{q}})$ the m-th column of $\boldsymbol{K}_{{n}}(\Dot{\boldsymbol{Q}}, \Dot{\boldsymbol{q}})$.
    Under Assumption~\ref{ass: class M full}, we can obtain
    \begin{equation}
    \begin{aligned}
        \|\tilde{\bm{k}}_{mn}\| &= \|\boldsymbol{k}^{\top}_{d_{mn}}(\Dot{\boldsymbol{Q}}, \Dot{\boldsymbol{q}})\Tilde{\dot{\boldsymbol{Q}}}_n\| \leq \|\boldsymbol{k}_{d_{mn}}(\Dot{\boldsymbol{Q}}, \Dot{\boldsymbol{q}})\|\|\Tilde{\dot{\boldsymbol{Q}}}_n\| \leq \sqrt{\mathcal{D}} |\sigma_{{f}_{mn}}| \|\Tilde{\dot{\boldsymbol{q}}}\|_{\infty}.
    \end{aligned}
    \end{equation}
    Through the matrix norm inequality \citep{LA_1983},
    \begin{equation}
    \begin{aligned}
        \|\Delta \boldsymbol{x}\| &=\|\boldsymbol{K}_{{\boldsymbol{y}}}^{-1}(\Dot{\boldsymbol{Q}}, \Dot{\boldsymbol{Q}})\Delta \boldsymbol{y}\| \leq \|\boldsymbol{K}_{{\boldsymbol{y}}}^{-1}(\Dot{\boldsymbol{Q}}, \Dot{\boldsymbol{Q}})\|\|\Delta \boldsymbol{y}\| 
        &\leq \sigma_{\max}\left(\boldsymbol{K}_{{\boldsymbol{y}}}^{-1}(\Dot{\boldsymbol{Q}}, \Dot{\boldsymbol{Q}})\right) \left\Vert\Delta \boldsymbol{y} \right\Vert \\
    \end{aligned}
    \end{equation}
    where $\sigma_{\max}$ denotes the maximum singular value of a matrix.
    Since $\boldsymbol{K}_{{\boldsymbol{y}}}^{-1}(\Dot{\boldsymbol{Q}}, \Dot{\boldsymbol{Q}})$ is positive definite and symmetric, its eigenvalues are equal to its singular values.
    Based on Weyl's inequality \citep{franklin_2012_matrix} and positive semidefiniteness of the Gram matrix,
    \begin{equation}
    \begin{aligned}
        \lambda_{\min}\left(\boldsymbol{K}_{\boldsymbol{y}}\right) = \lambda_{\min}\left(\boldsymbol{K}_{\boldsymbol{\tau}} + \sigma_\varepsilon^2 \boldsymbol{I}\right) \geq \lambda_{\min}\left(\boldsymbol{K}_{\boldsymbol{\tau}}\right) + \lambda_{\min}\left(\sigma_\varepsilon^2 \boldsymbol{I}\right)  \geq \lambda_{\min}\left(\sigma_\varepsilon^2 \boldsymbol{I}\right) = \sigma_\varepsilon^2.
    \end{aligned}
    \end{equation}
    We neglect the variable dependency here for notation simplicity.
    Since
    $
        \lambda_{\max}\left(\boldsymbol{A}^{-1}\right) = \frac{1}{\lambda_{\min}\left(\boldsymbol{A}\right)}
    $
    for regular $\bm{A}\in\mathbb{R}^{N\times N}$ then
    \begin{equation}
    \begin{aligned}
         \|\Delta \boldsymbol{x}\| &\leq  \frac{1}{\lambda_{\min}\left(\boldsymbol{K}_{\boldsymbol{y}}(\Dot{\boldsymbol{Q}}, \Dot{\boldsymbol{Q}})\right)} \left\Vert\Delta \boldsymbol{y} \right\Vert \leq \frac{ \left\Vert\Delta \boldsymbol{y} \right\Vert}{\sigma_\varepsilon^2} \\
        \tilde{\bm{k}}_{mn}^{\top}\Delta \boldsymbol{x} &\leq  |\tilde{\bm{k}}_{mn}^{\top}\Delta \boldsymbol{x}| \leq \|\tilde{\bm{k}}_{mn}\|\|\Delta \boldsymbol{x}\| \leq \sqrt{\mathcal{D}} \|\Tilde{\dot{\boldsymbol{q}}}\|_{\infty}\left\Vert\Delta \boldsymbol{y}\right\Vert \frac{|\sigma_{f_m}|}{\sigma_{\varepsilon}^2}.
    \end{aligned}
    \end{equation}
    Next, we derive from $\Dot{\boldsymbol{q}}^{\top}\!\hat{\boldsymbol{\tau}}^d\!\geq\!\bm{0}$ the inequality $\bm{M_D}\!\succeq\![\, \tilde{\bm{k}}_{mn}^\top \Delta\bm{x} \,]_{mn}$. Broken up further, this gives
    \begin{equation}
    \begin{aligned}
        \Dot{\boldsymbol{q}}^{\top}\hat{\boldsymbol{\tau}}^d(\Dot{\boldsymbol{q}}) &\geq \Dot{\boldsymbol{q}}^{\top}\mathrm{diag}(\bm{m_d}(\dot{\bm{q}}))\Dot{\boldsymbol{q}} - \sum_{m=1}^{N} \sum_{n=1}^{N} | \Dot{q}_m \Dot{q}_n \tilde{\bm{k}}_{mn}^{\top}\Delta \boldsymbol{x} | \\
        & \geq \Dot{\boldsymbol{q}}^{\top}\mathrm{diag}(\bm{m_d}(\dot{\bm{q}}))\Dot{\boldsymbol{q}} - \sum_{m=1}^{N} \sum_{n=1}^{N} | \Dot{q}_m | | \Dot{q}_n| \|\tilde{\bm{k}}_{mn}\|\|\Delta \boldsymbol{x}\| \\
        & \geq \Dot{\boldsymbol{q}}^{\top}\mathrm{diag}(\bm{m_d}(\dot{\bm{q}}))\Dot{\boldsymbol{q}} - \sqrt{\mathcal{D}} \|\Tilde{\dot{\boldsymbol{q}}}\|_{\infty}\left\Vert\Delta \boldsymbol{y}\right\Vert \frac{1}{\sigma_{\varepsilon}^2} \sum_{m=1}^{N} \sum_{n=1}^{N} | \Dot{q}_m | | \Dot{q}_n|  |\sigma_{f_{mn}}| \\
        & = |\Dot{\boldsymbol{q}}|^{\top}\left(\mathrm{diag}(\bm{m_d}(\dot{\bm{q}})) - \tfrac{\sqrt{\mathcal{D}} \|\Tilde{\dot{\boldsymbol{q}}}\|_{\infty}\left\Vert\Delta \boldsymbol{y}\right\Vert}{\sigma_{\varepsilon}^2} \boldsymbol{\Sigma}_f\right)|\dot{\bm{q}}| \geq 0
    \end{aligned}
    \end{equation}
    Thus, passivity can be guaranteed by enforcing the matrix inequality~\eqref{eq: full D-GP bound}, proving the result.
\end{proof}
Theorem~\ref{th: full bound} asserts that, when provided with observations of velocities and torques, ensuring guaranteed passivity involves selecting a basic kernel component, such as the SE kernel, whose bound can be regulated through hyperparameters $\sigma_{f_{mn}}$. This guarantee can be achieved by opting for either a sufficiently large prior mean or noise variance. However, the bound of~\eqref{eq: full D-GP bound} becomes tighter for increasing data set sizes.
A specific conclusion can also be drawn when we contemplate a diagonal damping matrix.

\subsection{Matrix-vector Structured Gaussian Processes for Diagonal Damping Matrices}
\label{sec: Diag-D-GP}

\subsubsection{GP Structure and Formulation}
In this section, we consider a specific restriction on the damping matrix where
    $\boldsymbol{D}(\Dot{\boldsymbol{q}}) \in \mathbb{R}^{N \times N}$ is a diagonal damping matrix which can be written as $\boldsymbol{D}(\Dot{\boldsymbol{q}}) = \text{diag}(d_1(\Dot{\boldsymbol{q}}), d_2(\Dot{\boldsymbol{q}}), \ldots, d_N(\Dot{\boldsymbol{q}}))$.
Defining $\boldsymbol{d}(\Dot{\boldsymbol{q}}) = [d_1(\Dot{\boldsymbol{q}}), d_2(\Dot{\boldsymbol{q}}), \ldots, d_N(\Dot{\boldsymbol{q}})]^{\top}$, a vector-valued damping function can be modeled as
\begin{equation}
    \boldsymbol{d}(\Dot{\boldsymbol{q}}) \sim \mathcal{G} \mathcal{P}\left(\bm{M_D}(\Dot{\boldsymbol{q}}), \boldsymbol{K}_{\boldsymbol{d}}(\Dot{\boldsymbol{q}}, \Dot{\boldsymbol{q}}^{\prime} )\right),
    \label{eq: diag basic GP}
\end{equation}
where 
$\bm{M_D}(\Dot{\boldsymbol{q}}) = [m_{d_1}(\Dot{\boldsymbol{q}}), \dots, m_{d_N}(\Dot{\boldsymbol{q}})]^{\top}$
and
$
\boldsymbol{K}_{\boldsymbol{d}}(\Dot{\boldsymbol{q}}, \Dot{\boldsymbol{q}}^{\prime}) = \mathrm{diag}(
        k_{d_1}\left(\Dot{\boldsymbol{q}}, \Dot{\boldsymbol{q}}^{\prime}\right), \dots k_{d_N}\left(\Dot{\boldsymbol{q}}, \Dot{\boldsymbol{q}}^{\prime}\right)
        ).
$ 
We can write the dissipative torque as
\begin{equation}
    \boldsymbol{\tau}^d(\Dot{\boldsymbol{q}}) = \boldsymbol{D}(\Dot{\boldsymbol{q}})\Dot{\boldsymbol{q}} = \text{diag}(\boldsymbol{d}(\Dot{\boldsymbol{q}}))\Dot{\boldsymbol{q}} = \text{diag}(\Dot{\boldsymbol{q}}) \boldsymbol{d}(\Dot{\boldsymbol{q}}),
\end{equation}
Defining a mean matrix $\bm{M_D}(\Dot{\boldsymbol{q}}) = \operatorname{diag}(\bm{M_D}(\Dot{\boldsymbol{q}}))$, 
the covariance of the damping torque $\boldsymbol{\tau}^d(\Dot{\boldsymbol{q}})$ is
\begin{equation}
\begin{aligned}
        \operatorname{cov}(\boldsymbol{\tau}^d(\Dot{\boldsymbol{q}})) &= \mathbb{E}[(\boldsymbol{D}(\Dot{\boldsymbol{q}})\Dot{\boldsymbol{q}} - \bm{M_D}(\Dot{\boldsymbol{q}})\Dot{\boldsymbol{q}})(\boldsymbol{D}(\Dot{\boldsymbol{q}}^{\prime})\Dot{\boldsymbol{q}}^{\prime} - \bm{M_D}(\Dot{\boldsymbol{q}}^{\prime})\Dot{\boldsymbol{q}}^{\prime})^{\top}] \\
        &= \mathbb{E}[\operatorname{diag}(\Dot{\boldsymbol{q}})(\boldsymbol{d}(\Dot{\boldsymbol{q}})-\bm{M_D}(\Dot{\boldsymbol{q}}))(\boldsymbol{d}(\Dot{\boldsymbol{q}}^{\prime})-\bm{M_D}(\Dot{\boldsymbol{q}}^{\prime}))^{\top}\operatorname{diag}(\Dot{\boldsymbol{q}}^{\prime})^{\top}] \\
        &= \operatorname{diag}(\Dot{\boldsymbol{q}})\mathbb{E}[(\boldsymbol{d}(\Dot{\boldsymbol{q}})-\bm{M_D}(\Dot{\boldsymbol{q}}))(\boldsymbol{d}(\Dot{\boldsymbol{q}}^{\prime})-\bm{M_D}(\Dot{\boldsymbol{q}}^{\prime}))^{\top}]\operatorname{diag}(\Dot{\boldsymbol{q}}^{\prime}) \\
        &=\operatorname{diag}(\Dot{\boldsymbol{q}})
        \boldsymbol{K}_{\boldsymbol{d}}(\Dot{\boldsymbol{q}}, \Dot{\boldsymbol{q}}^{\prime})
        \operatorname{diag}(\Dot{\boldsymbol{q}}^{\prime})
\end{aligned}
\end{equation}

\noindent The torque GP model and joint distribution are given analogously by~\eqref{eq: Diag tau GP}--\eqref{eq: Diag-D-GP}
with different kernel
\begin{equation}
   \boldsymbol{K}_{\boldsymbol{\tau}}(\Dot{\boldsymbol{q}}, \Dot{\boldsymbol{q}}^{\prime})=\operatorname{diag}(\Dot{\boldsymbol{q}})
        \boldsymbol{K}_{\boldsymbol{d}}(\Dot{\boldsymbol{q}}, \Dot{\boldsymbol{q}}^{\prime})
        \operatorname{diag}(\Dot{\boldsymbol{q}}^{\prime}) \,.
\end{equation}

\subsubsection{Positive Semidefiniteness}

Having introduced this diagonal structured GP formulation, we investigate its physical properties as positive semidefiniteness of damping matrix to guarantee passivity.
Before that, We also first impose a trivial structural restriction on the used kernel functions.
\begin{assumption}
    Each 
    kernel component $k_{dn}(\Dot{\boldsymbol{q}}, \Dot{\boldsymbol{q}}^{\prime})$ is bounded by $\sigma_{f_n}^2$,
    i.e. $|k_{d_{n}}\left(\Dot{\boldsymbol{q}}, \Dot{\boldsymbol{q}}^{\prime}\right)| \leq \sigma_{f_{n}}^2$.
\label{ass: class M}
\end{assumption}

With the diagonal form of $\boldsymbol{D}(\Dot{\boldsymbol{q}})$, the torque components are independent, so w.l.o.g. by marginalization, we can analyze one component of the torque estimation to show its property, which is
    \begin{equation}
    \begin{aligned}
        \hat{\tau}_n^d(\Dot{\boldsymbol{q}}) 
        &= m_{d_n}(\Dot{\boldsymbol{q}})\Dot{q}_n + \Dot{q}_n \boldsymbol{k}^{\top}_{d_n}(\Dot{\boldsymbol{Q}}, \Dot{\boldsymbol{q}})\operatorname{diag}(\Tilde{\dot{{\boldsymbol{q}}}})\boldsymbol{K}_{\boldsymbol{y}}^{-1}(\Dot{\boldsymbol{Q}}, \Dot{\boldsymbol{Q}})(\boldsymbol{y} - \boldsymbol{m_y})
    \label{eq: tau n hat}
    \end{aligned}
    \end{equation}
    where $\boldsymbol{k}_{\tau_n}(\Dot{\boldsymbol{Q}}, \Dot{\boldsymbol{q}})$ is the n-th colunm in $\boldsymbol{K_{\tau}}(\Dot{\boldsymbol{Q}}, \Dot{\boldsymbol{q}})$, and $\boldsymbol{k}^{\top}_{d_n}(\Dot{\boldsymbol{Q}}, \Dot{\boldsymbol{q}})$ is the n-th colunm in $\boldsymbol{K_{d}}(\Dot{\boldsymbol{Q}}, \Dot{\boldsymbol{q}})$.
\begin{corollary_}
    Consider the posterior estimate $\hat{\boldsymbol{D}}(\Dot{\boldsymbol{q}})$ with kernels components under Ass.~\ref{ass: class M} and the damping torque estimate~\eqref{eq: tau n hat} for each torque component. 
    Positive semidefiniteness of the matrix estimate $\hat{\boldsymbol{D}}(\Dot{\boldsymbol{q}})$ is guaranteed deterministically if $\forall n=1, \dots, N$, 
    \begin{equation}
        |\sigma_{f_n}| \leq \frac{\sigma_{\varepsilon}^2}{\sqrt{D} \|\Tilde{\dot{\boldsymbol{q}}}\|_{\infty} \left\Vert\Delta \boldsymbol{y} \right\Vert}m_{d_n}(\Dot{\boldsymbol{q}})
    \label{eq: Diag-D-GP bound}
    \end{equation}
    \label{th: diag bound}
\end{corollary_}
\begin{proof}
    According to~\eqref{eq: tau n hat}, the dissipated power for each dimension is
    \begin{equation}
    \Dot{q}_n\hat{\tau}_n^d(\Dot{\boldsymbol{q}}) = {\Dot{q}_n}^2(m_{d_n}(\Dot{\boldsymbol{q}}) + \underbrace{\boldsymbol{k}^{\top}_{d_n}(\Dot{\boldsymbol{Q}}, \Dot{\boldsymbol{q}})\operatorname{diag}(\Tilde{\dot{{\boldsymbol{q}}}})}_{\boldsymbol{k}_n^{\top}}\underbrace{\boldsymbol{K}_{\boldsymbol{y}}^{-1}(\Dot{\boldsymbol{Q}}, \Dot{\boldsymbol{Q}})\Delta \boldsymbol{y}}_{\Delta \boldsymbol{x}}).
    \end{equation}
    Similar as the proof in Theorem~\ref{th: full bound}, We can obtain
    \begin{equation}
    \begin{aligned}
        \|\boldsymbol{k}_n\| &= \|\boldsymbol{k}^{\top}_{d_n}(\Dot{\boldsymbol{Q}}, \Dot{\boldsymbol{q}})\operatorname{diag}(\Tilde{\dot{{\boldsymbol{q}}}})\| \leq \|\boldsymbol{k}_{d_n}(\Dot{\boldsymbol{Q}}, \Dot{\boldsymbol{q}})\|\|\operatorname{diag}(\Tilde{\dot{{\boldsymbol{q}}}})\| \leq \sqrt{\mathcal{D}} |\sigma_{f_n}| \|\Tilde{\dot{\boldsymbol{q}}}\|_{\infty} \\
    \end{aligned}
    \end{equation}
    \begin{equation}
    \begin{aligned}
        \|\Delta \boldsymbol{x}\| &=\|\boldsymbol{K}_{{\boldsymbol{y}}}^{-1}(\Dot{\boldsymbol{Q}}, \Dot{\boldsymbol{Q}})\Delta \boldsymbol{y}\| \leq \|\boldsymbol{K}_{{\boldsymbol{y}}}^{-1}(\Dot{\boldsymbol{Q}}, \Dot{\boldsymbol{Q}})\|\|\Delta \boldsymbol{y}\| 
        \leq \frac{\left\Vert\Delta \boldsymbol{y}\right\Vert}{\sigma_{\varepsilon}^2} \\
    \end{aligned}
    \end{equation}
    then
    \begin{equation}
    \begin{aligned}
        \boldsymbol{k}_n^{\top}\Delta \boldsymbol{x} &\geq - |\boldsymbol{k}_n^{\top}\Delta \boldsymbol{x}| \geq -\|\boldsymbol{k}_n\|\|\Delta \boldsymbol{x}\| \geq -\sqrt{D} \|\Tilde{\dot{\boldsymbol{q}}}\|_{\infty}\left\Vert\Delta \boldsymbol{y} \right\Vert \frac{|\sigma_{f_n}|}{\sigma_\varepsilon^2}
    \end{aligned}
    \end{equation}
    With the condition of~\eqref{eq: Diag-D-GP bound}, 
    $
        \boldsymbol{k}_n^{\top}\Delta \boldsymbol{x} \geq - m_{d_n}(\Dot{\boldsymbol{q}})
    $ and
    the dissipated power
    \begin{equation}
    \Dot{q}_n\hat{\tau}_n^d(\Dot{\boldsymbol{q}}) = {\Dot{q}_n}^2\left(m_{d_n}(\Dot{\boldsymbol{q}}) + \boldsymbol{k}_n^{\top}\Delta \boldsymbol{x} \right) \geq 0.
    \end{equation}
    Because of the diagonal structure of $\hat{\boldsymbol{D}}(\Dot{\boldsymbol{q}})$, its positive semidefiniteness can be directly derived.
\end{proof}

\section{Simulation and Results}
\label{sec: Simulation and Results}

We evaluate our methods by estimating the unknown aerodynamics based on an aircraft benchmark by ONERA \& AIRBUS, which simulates the aircraft landing period in full configuration from 1000ft above the runway until touch-down \citep{CALC_2016}. 
We determine the domain of the air velocities by setting up the three-dimensional wind inputs that change periodically within $[-25, 25] m/s$ with different frequencies $[0,1, 0.2, 0.3]$ and phases $[0, 2, 3]$, and obtain the air velocities in an angle-based coordinate which consists of the angle of attack $\alpha$, the side-slip angle $\beta$ and the norm of the airspeed $v_a$. The generalized velocity is set to $\Dot{\boldsymbol{q}} = [\alpha, \beta, v_a]^\top$. We also use the benchmark model to collect the corresponding aerodynamic forces $\boldsymbol{\tau}^d$, namely drag, lateral and lift. For notation simplicity, we use ARD-GP to denote the standard GP methods with SE kernel based on automatic relevance determination (ARD) \citep{rasmussen_gaussian_2006}, Diag-D-GP to denote the GP methods with the  by our proposed diagonal form damping matrix, and Full-D-GP to denote the methods with the full damping matrix.
To make comparison more in consistent and the optimization process simpler, we choose a fixed noise variance $\sigma_\varepsilon=100$ for all unconstrained method and choose $\sigma_\varepsilon=10^8$ to fulfill the constraint. Note that the aerodynamic forces are in the order of $10^6$, so this large noise variance is still reasonable in practice. We first optimize all the other hyperparameters of ARD-GP by minimizing the mean squared error on a validation set. Then we choose the same length-scale $\boldsymbol{\ell} = [18, 18, 0.2]^{\top}$ and optimize the hypervariances for all the other methods.
We first compare the estimation performance of each method and then show the guarantee of passivity numerically with our constrained GP.

\subsection{Tracking Performance and Data Efficiency}

The training results are evaluated in two types. One is the aerodynamics tracking error based on a trajectory along the landing period, as shown in Fig.~\ref{fig: fitting traj error}. 
We sample 50 data points from the trajectory with equivalent interval as the training set and 100 points as validation set. 
The relative error is the difference between the estimation and ground truth normalized with the mass of the aircraft. Our proposed Full-D-GP performs the best compared with the others, which obtain a relative error with almost 0 for all three dimensions.
The statistical values of the lateral estimation error as an example for different methods are listed in Table~\ref{tab: statistical evaluation}.
Note that with the assumption of a diagonal damping matrix, we lose the dependencies between each dimension. The relative error of drag in Diag-D-GP is larger than the others because the drag should also contain components based on the last velocity dimension, norm of the airspeed.
The other evaluation type is in a volume space in the domain of the three-dimensional inputs. In this case, the inputs can be chosen as any combination of the values for each dimension in the limited domain. We evaluate the normalized mean squared error (NMSE) in the volume space with increased training data by
$
    {\sum_{i=1}^{N}(\hat{y}_i - y_i)^2}/{\sum_{i=1}^{N}(y_i - \bar{y})^2},
$
where $\hat{y}_i$ and $y_i$ are the estimation and ground truth of the i-th data, respectively, and $\bar{y}$ is the mean of the data. The results are shown in Fig~\ref{fig: training efficiency}. Full-D-GP achieves the highest training efficiency for all dimensions and converges to a lower NMSE, which follows the same phenomena as training for a trajectory.
The NMSE estimated by the constrained Full-D-GP remains at a certain level with the increase of the data set because of the large noise variance.
The simulation results show a significant improvement in data efficiency of training when we consider more physical structure as a prior. By involving our proposed constraint, the Full-D-GP maintains efficiency and reasonable performance.

\begin{figure}[htb]
\centering
\subfigure[Unit mass error of aerodynamics estimation for a trajectory]{
  \includegraphics[width=0.46\textwidth] {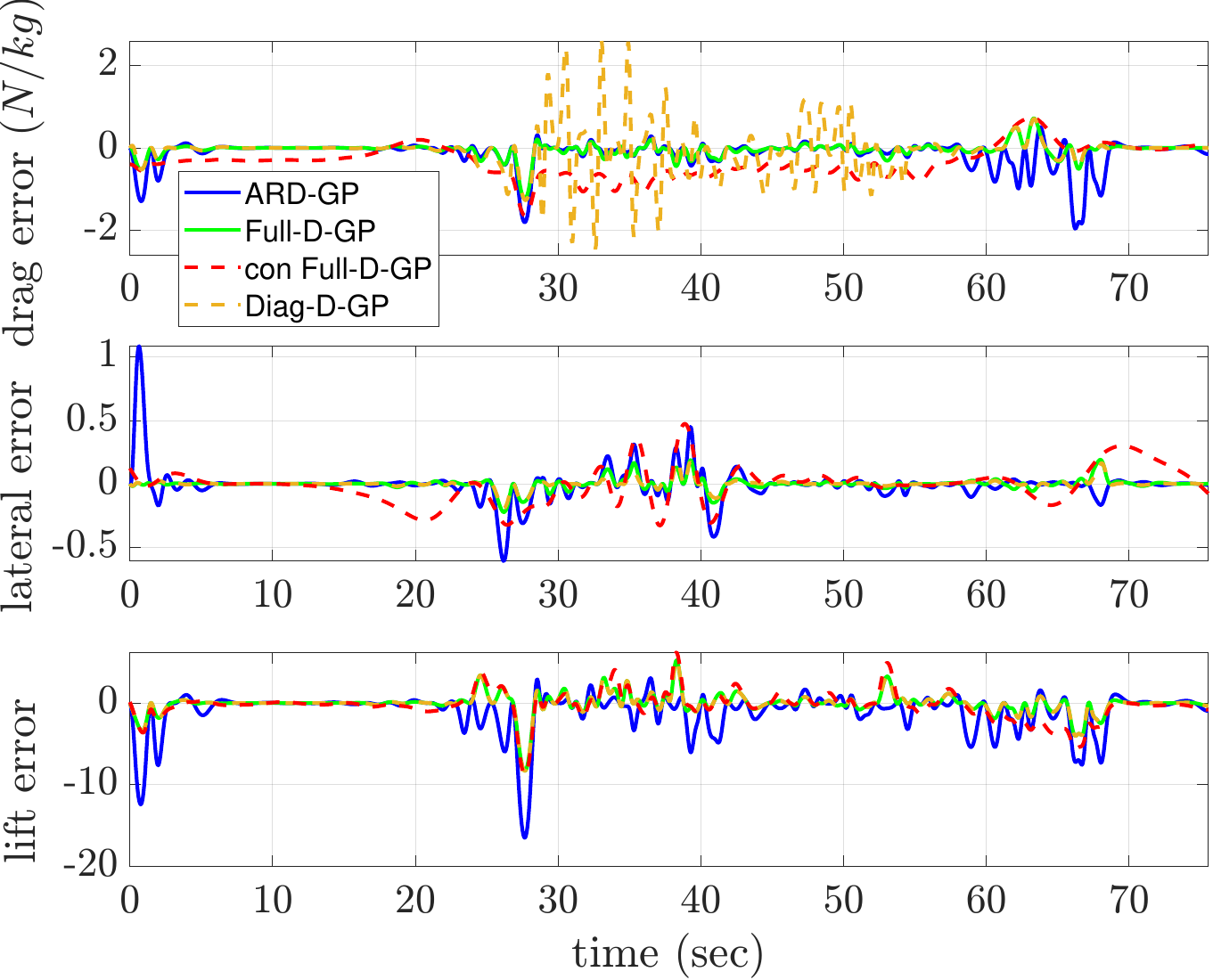}
  \label{fig: fitting traj error}
 }
\quad
\subfigure[Data efficiency in an uniformly-spaced airspeed volume]{
  \includegraphics[width=0.47\textwidth] {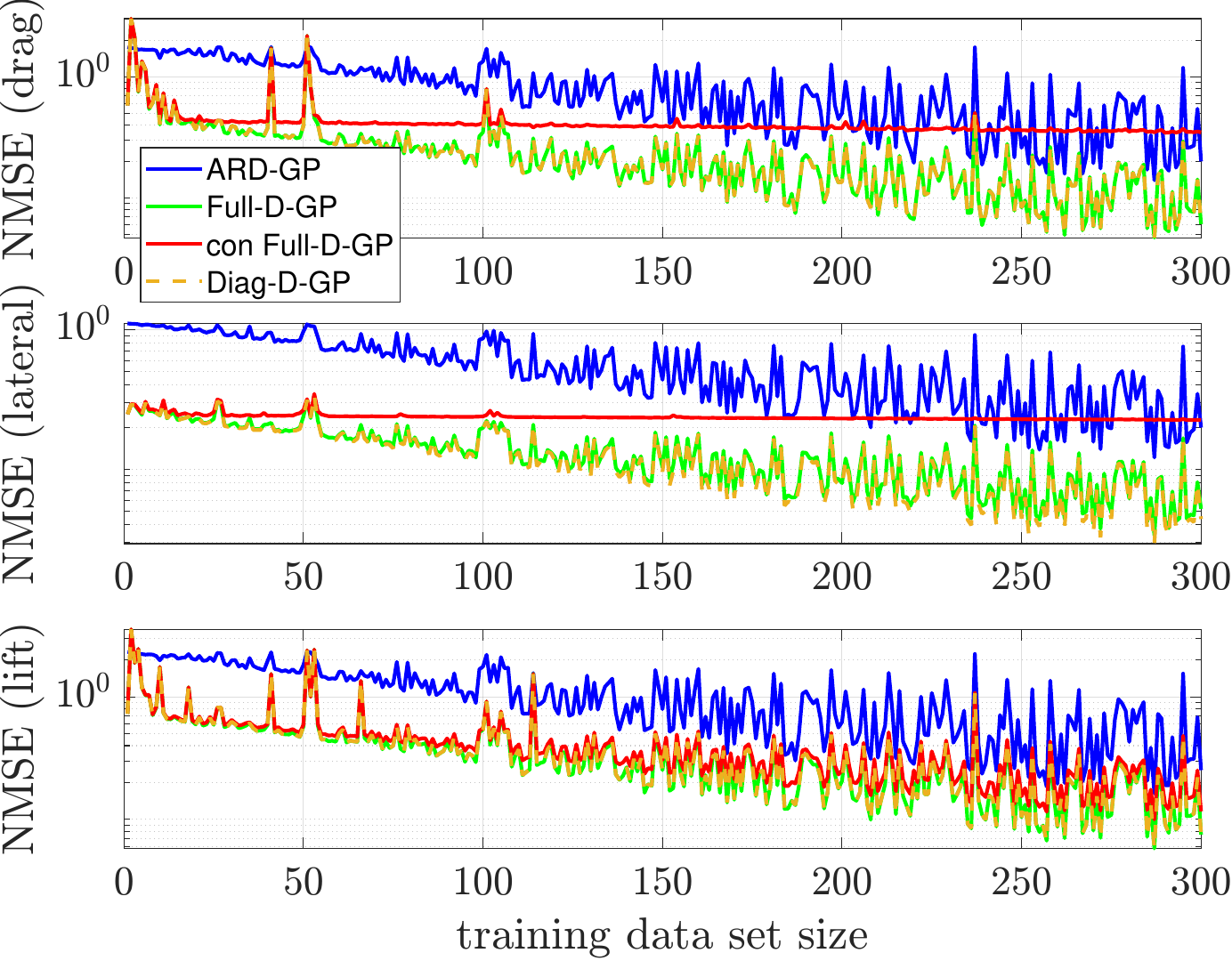}
  \label{fig: training efficiency}
 }
\caption{Training results based on trajectory and volume space}
\label{fig: Fitting performance}
\end{figure}

\begin{table}[htb]
    \centering
    \begin{tabular}{cccc}
         & Full-D-GP & Diag-D-GP & ARD-GP\\
        \hline
        mean & $-7.83 \times 10^{-4}$ & $-5.73 \times 10^{-4}$  & -0.0052 \\
        \hline
        variance & 0.0019 & 0.0012 & 0.175 \\
    \end{tabular}
    \caption{Statistical values of the relative lateral force estimation error along a landing trajectory}
    \label{tab: statistical evaluation}
\end{table}

\subsection{Passivity Characteristic}
The passivity can be preserved if the hyperparameters are chosen with the conditions we derive in Section~\ref{sec: Dissipative Gaussian Processes}. By involving these conditions as constraints during optimization, we can show the passivity of the estimation numerically by presenting 
the dissipated power.
Consider the sufficient condition~\eqref{eq: full D-GP bound} for full structured kernel, we determine the mean $m_{d_n}(\dot{\boldsymbol{q}})$ by solving
a least squared problem to minimize $\|\boldsymbol{\tau}_n^d - m_{d_n} (\dot{\boldsymbol{q}})\dot{\boldsymbol{q}}_n\|$ where $\boldsymbol{\tau}_n^d$ and $\dot{\boldsymbol{q}}_n$ are the collected forces and velocities of n-th dimension, respectively.
$\mathcal{D}, \|\Tilde{\dot{\boldsymbol{q}}}\|_{\infty}$ and $\left\Vert\Delta \boldsymbol{y}\right\Vert$ can be directly obtained when the data is given. 
The dissipated power $P = \dot{\boldsymbol{q}}^{\top} \boldsymbol{\tau}^d$ should always be non-negative for a passive system, whose distribution based on our test data collected from the aircraft benchmark
and the corresponding estimation based on the hyperparameters with and without constraints are shown in Fig.~\ref{fig: Inputs-outputs mapping Diag-D-GP}.
The collected data is noisy but most of the points are positive, so the estimation based on non-constrained GP shifts the distribution to a positive direction but still cannot preserve passivity. The passivity is deterministically guaranteed with constrained hyperparameters.

\begin{figure}[htb]
    \centering
    \includegraphics[width=0.9\textwidth]{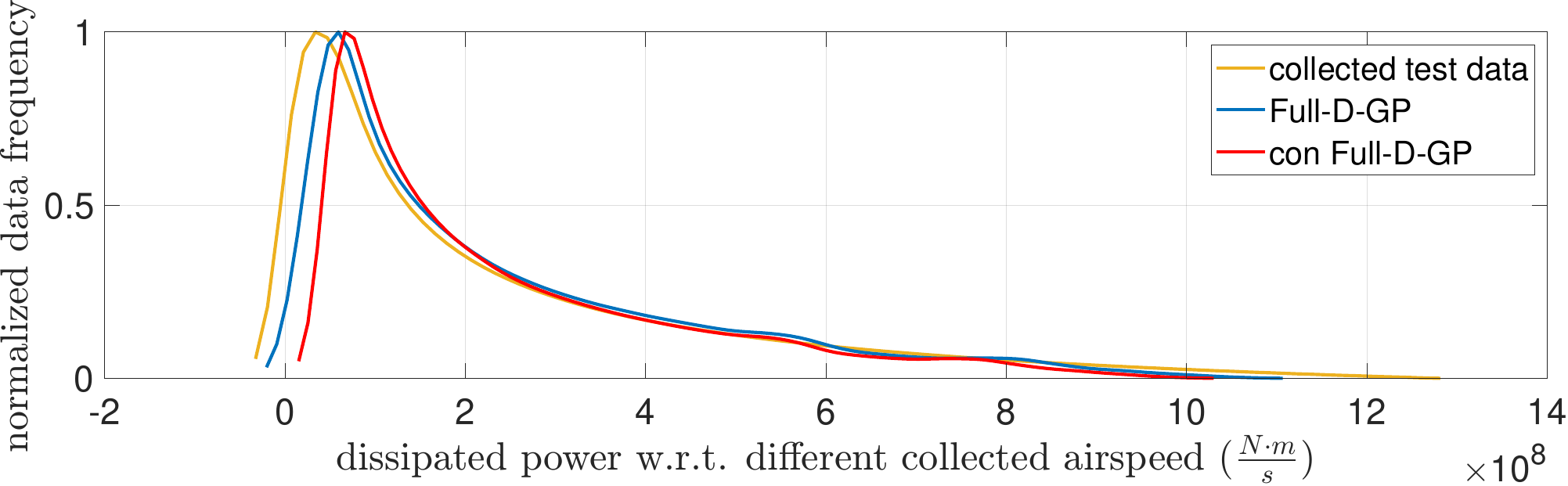}
    \caption{Dissipated power distribution from the collected test data and the corresponding estimation based on constrained and unconstrained methods}
    \label{fig: Inputs-outputs mapping Diag-D-GP}
\end{figure}

\section{Conclusion}
\label{sec: Conclusion}

Our approach towards structural learning methods for friction modeling and identification integrates physical structure as prior knowledge with Gaussian Processes. Considering the mechanical system with the velocity-dependent dissipative term, we introduce a matrix-vector structure for prior mean and kernels, with a full and diagonal damping matrix, respectively. We investigate their passivity guarantee of the damping estimates by showing the positive semidefiniteness of the damping matrix. Constrained Gaussian Process algorithms are proposed to guarantee passivity during training by optimizing hyperparameters with constraints. An aircraft benchmark simulation is employed to demonstrate the efficacy of our methodology. By considering more information in the structure, higher data efficiency and estimate accuracy can be obtained. In addition, the passivity characteristic of the estimated aerodynamic forces w.r.t. the generalized velocities is validated numerically.

\acks{This work has received support by the Horizon Europe Framework Programme project euROBIN and by the European Research Council (ERC) Consolidator Grant "Safe data-driven control for human-centric systems" (CO-MAN) under grant agreements 101070596 and 864686, respectively.}

\bibliographystyle{ieeetr}

\end{document}